\documentclass[12pt]{article}
\usepackage{amsmath}
\usepackage{graphicx,psfrag,epsf}
\usepackage{enumerate}
\usepackage{amsfonts}
\usepackage{booktabs}
\usepackage{amsmath,epsfig,amssymb,amsfonts,amsthm,epsfig,verbatim}
\usepackage{color,graphicx,lscape,longtable}
\usepackage{float}
\usepackage{subfigure}
\usepackage{indentfirst}
\usepackage{epstopdf}
\usepackage{xcolor}
\usepackage[round,authoryear]{natbib}
\usepackage{threeparttable}
\usepackage{subfigure}
\usepackage{bm}
\usepackage{rotating}

\def\vc{\text{VC}}

\def\hat{\widehat}
\def\tilde{\widetilde}

\newtheorem{Theorem}{Theorem}

\newtheorem{Corollary}[Theorem]{Corollary}

\begin{document}

\def\spacingset#1{\renewcommand{\baselinestretch}%
{#1}\small\normalsize} \spacingset{1}

%%%%%%%%%%%%%%%%%%%%%%%%%%%%%%%%%%%%%%%%%%%%%%%%%%%%%%%%%%%%%%%%%%%%%%%%%%%%%%

\title{\bf  Instrument, Variable and Model Selection with Nonignorable Nonresponse}

\author{Ji Chen$^{1}$ and Jun Shao${^{1,2}}$\\
$^{1}$School of Statistic, East China Normal University\\
$^{2}$Department of Statistics, University of Wisconsin-Madison}
\date{}

\maketitle

\spacingset{1.5}

\begin{center}
	{\bf Abstract}
\end{center}

With nonignorable nonresponse, an effective method to construct valid estimators of population parameters is to use a covariate vector called instrument that can be excluded from the nonresponse propensity but are still useful covariate even when other covariates are conditioned. The existing work in this approach assumes such an instrument is given, which is frequently not the case in applications.
In this paper we investigate how to search for an instrument from a given set of covariates. The method for estimation we apply is the pseudo likelihood proposed by
\citet{tang2003analysis} and \citet{zhao2015semiparametric}, which assumed that
an instrument is given and the distribution of response given covariates is parametric and the propensity is nonparametric. Thus, in addition to the challenge of searching an instrument,  we also need to do variable and model selection simultaneously.
We propose a method for instrument, variable, and model selection and show that our method produces consistent instrument and model selection as the sample size tends to infinity, under some regularity conditions.
Empirical results including two simulation studies and two real examples are present to show that the proposed method works well.
\vspace{2mm}

\noindent%

\newpage

\section{Introduction\label{sec:intro}}
Nonresponse with an appreciable rate is common in many applications such as clinical trials and sample surveys.
Let $Y$ be a response or outcome of interest that may have nonresponse,  $X$ be a $p$-dimensional covariate vector that is always observed, and $R$ be the  indicator equaling  1 if $Y$ is observed and 0 if $Y$ is missing. When the propensity $P(R=1|Y,X)$ is equal to $P(R=1|X)$ not depending on $Y$, the nonresponse is called ignorable
and there is a  rich literature on methodology of handling nonresponse \citep{little2002statistical}. In many applications,  $P(R=1|Y,X)$ depends on both $X$ and $Y$,  in which cases nonresponse is referred to as nonignorable and  estimation of population parameters is much more challenging than that in the case of ignorable nonresponse.

Throughout we use $p( \cdot  | \cdot )$ or $p( \cdot )$ as a generic notation for the conditional or unconditional probability density with respect to an appropriate measure (discrete, continuous, or mixed). With nonignorable nonresponse, when both $p(Y|X)$ and $P(R=1|Y,X)$ are parametric, maximum likelihood methods have been developed \citep{greenlees1982imputation,baker1988regression}. When both $p(Y|X)$ and $P(R=1|Y,X)$ are nonparametric, \citet{robins1997toward} showed that the population may not be  identifiable. Hence, efforts have been made to develop semiparametric methods, assuming one of $p(Y|X)$ and $P(R=1|Y,X)$ has a parametric form and the other one is nonparametric.
\citet{qin2002estimation} and \citet{wang2014instrumental} imposed a parametric model on $P(R=1|Y,X)$ but allowed $p(Y|X)$ to be nonparametric. Following \citet{tang2003analysis},  \citet{zhao2015semiparametric}, and \citet{chen2018pseudo}, in this paper we focus on
 a nonparametric $P(R=1|Y,X)$ and a parametric model
 \begin{equation}
 p(Y|X)=p(Y|X;\theta),\label{para}
 \end{equation}
 where  $\theta$ is an unknown parameter vector and $p(Y|X;\theta)$ is known when $\theta$ is known.

In their semiparametric approach, \cite{zhao2015semiparametric} utilized a covariate vector $Z$ called nonresponse instrument or simply instrument,
to guarantee the identifiability of population parameters so that consistent estimators can be obtained. More precisely, an instrument $Z$ is a sub-vector of $X$, i.e.,  $X=(U,Z)$, such that $Z$ satisfies the following two conditions:
\begin{align}
&P(R=1 |Y,X)=P(R=1|Y,U), \label{eq:propensity}\\
&p(Y|X)=p(Y|U,Z) \text{ depends on } Z. \label{eq:propensity1}
\end{align}
If we know which components of $X$ satisfy (\ref{eq:propensity})-(\ref{eq:propensity1}), then parameters in $p(Y|X)$ can be estimated using pseudo likelihoods \citep{zhao2015semiparametric} and parametric model selection regarding (\ref{para}) can also be performed \citep{Fang2016}.
Note that even if both $p(Y|X)$ and $P(R=1|Y,X)$ are parametric, there is still an identifiability issue and the use of an instrument satisfying (\ref{eq:propensity})-(\ref{eq:propensity1}) may be needed.

In most applications, however, an instrument is not given and we must search for an instrument using observed data. \citet{chen2021instrument} conducted an initial study on instrument selection using a pseudo-likelihood approach. The problem becomes even more challenging when we also need to remove redundant components of  $X$ as well as to select models  for the parametric assumption in (\ref{para}).

The purpose of this paper is to propose and study a method for instrument search, dimension reduction, and parametric model selection simultaneously. After an introduction of a pseudo likelihood method in Section 2, we propose a pseudo likelihood based maximum ratio criterion to select an instrument, assuming that we do not need to do dimension reduction or model selection.
In  Section 3 we perform variable selection to eliminate redundant components of  the instrument selected in Section 2 and also redundant components in the entire covariate vector $X$. Section 4 is devoted to parametric model selection regarding (\ref{para}).
To complement theoretical properties of the proposed method derived in Sections 2-4, we carry out some simulations in Section 5 to examine finite sample properties. For illustration, we apply the proposed method to two data sets in Section 6. All technical details are given in an Appendix.

\section{Instrument Selection\label{sec:ins}}

We use the notation from Section \ref{sec:intro}: $Y$ is a response subject to nonresponse, $R$ is its observation indicator, and $X$ is a fully observed $p$-dimensional covariate vector. We assume the parametric model~\eqref{para} for $p(Y|X)$ is correct. Our goal is to estimate $\theta$ from a random sample $\{(y_i, x_i, r_i): i=1, \dots, N\}$.

We assume there exists an instrument $Z$, a sub-vector of $X$, satisfying conditions~\eqref{eq:propensity} and~\eqref{eq:propensity1}. If $Z$ were known, $\theta$ could be estimated by maximizing the pseudo-likelihood~\citep{tang2003analysis, zhao2015semiparametric}:
\begin{equation}
\prod_{i: r_i=1} p(z_i | y_i, u_i) = \prod_{i: r_i=1} \frac{p(y_i | u_i, z_i; \theta) p(u_i | z_i; \hat{\eta}) d\hat{F}(z_i)}{\int p(y_i | u_i, z; \theta) p(u_i | z; \hat{\eta}) d\hat{F}(z)} \label{likelihood}
\end{equation}
where $\hat{F}$ is the empirical CDF of $Z$, and $\hat{\eta}$ is an estimator for the parameters of $p(U|Z)$.

Since $Z$ is typically unknown, we propose a procedure to select it from $X$. The core idea, detailed in \citet{chen2021instrument}, is to compare two estimators of the conditional distribution $F_1(z) = P(Z \leq z | R=1)$ for a candidate $Z$:
\begin{itemize}
    \item The empirical CDF $\hat{F}_1(z)$, which is always consistent.
    \item A model-based estimator $\tilde{F}_1(z)$ derived from the pseudo-likelihood~\eqref{likelihood}.
\end{itemize}
If $Z$ is a valid instrument, these two estimators should be asymptotically close. We therefore define a validation criterion (VC) for a candidate instrument $Z_k$ as:
\begin{equation}
\text{VC}(k) = \frac{1}{N} \sum_{i=1}^{N} \left| \tilde{F}_{1k}(z_{ki}) - \hat{F}_{1}(z_{ki}) \right|. \label{eq:VC}
\end{equation}
A small $\text{VC}(k)$ indicates $Z_k$ likely satisfies the key exclusion condition~\eqref{eq:propensity}. However, $\text{VC}(k)$ cannot distinguish instruments that satisfy~\eqref{eq:propensity} but not~\eqref{eq:propensity1} (i.e., are redundant).

To overcome this, we first identify the set $S = \{k: Z_k \text{ satisfies~\eqref{eq:propensity}} \}$. The union $Z_S = \{Z_k: k \in S\}$ is then a valid instrument, as it satisfies~\eqref{eq:propensity} and must contain at least one component satisfying~\eqref{eq:propensity1}.

Let $\{l_1, \ldots, l_p\}$ be a permutation of $\{1, \ldots, p\}$ such that $\text{VC}(l_1) \leq \cdots \leq \text{VC}(l_p)$. We estimate the set $S$ by $\hat{S} = \{l_1, \ldots, l_{\hat{d}}\}$, where
\begin{equation}
\hat{d} = 
\begin{cases}
\arg\min_{1 \leq j \leq p-1} \text{VC}(l_j) / \text{VC}(l_{j+1}) & \text{if } \frac{\text{VC}(l_p) - \text{VC}(l_1)}{\text{VC}(l_1)} > (\log N)^{1/2} \\
p & \text{otherwise.}
\end{cases} \label{eq:d}
\end{equation}
This thresholding rule is motivated by the asymptotic behavior of the VC statistic \citep{huang2014feature}. As shown in \citet{chen2021instrument}, under standard regularity conditions, $P(\hat{S} = S) \to 1$ as $N \to \infty$.

Finally, as a single binary covariate often cannot serve as a valid instrument alone \citep{zhao2015semiparametric}, we pre-process the covariates: any binary component is combined with another covariate to form a vector candidate before applying the selection procedure~\eqref{eq:d}.

The selected instrument is $\hat{Z} = Z_{\hat{S}}$, and the remaining covariates form $\hat{U}$. The parameter $\theta$ can then be estimated by maximizing~\eqref{likelihood} using $\hat{Z}$ and $\hat{U}$.

\section{Variable and Instrument Selection\label{sec:var}}

It is shown in Section \ref{sec:ins} that, with probability tending to 1, we can find a vector $Z_S$ satisfying \eqref{eq:propensity}-\eqref{eq:propensity1} and having  the largest possible dimension. However, we need to further address the following two issues.
\vspace{-1mm}
\begin{itemize}
	\item[(a)] Although $Z_S$ defined in Section 2 satisfies \eqref{eq:propensity} and \eqref{eq:propensity1},
	some components of $Z_S$ may not satisfy \eqref{eq:propensity1},
	i.e., they are redundant and lead to less efficient estimators of $\theta$. \vspace{-3mm}
	\item[(b)]  	Some components of $X$ may be redundant and should be excluded from $p(Y|X)$.
\end{itemize}
\vspace{-1mm}

Assuming that model (\ref{para}) is correct, in this section we address (a)-(b),
i.e., we try to eliminate redundant components in $Z_S$ and  $p(Y|X)$.
 In what follows, for two vectors $A$ and $B$, $A \subset B$ means that $A$ is a sub-vector of $B$,
$A \cup B$ denotes the vector with components in either $A$ or $B$,
$A \cap B$ denotes the vector with components in both $A$ and $B$,
and
$A \cap B = \varnothing$ means that $A$ and $B$ have no common component.

%Now we first assume the parametric form $p(Y|X;\theta )$ for $p(Y|X)$ is correct, but $p(Y|X;\theta )$ could be an overfitted model, in which
Let $X_*\subset X$ be the sub-vector with the smallest dimension such that \begin{equation}
\label{bestcov}
p(Y|X)=p(Y|X_*),
\end{equation} i.e., components of $X$ not in $X_*$ are redundant in terms of predicting $Y$. Also, let
\begin{equation}
Z_S=(Z_*,V), \quad Z_*\subset X_*, \ V \cap X_* = \varnothing , \label{bestins}
\end{equation}
i.e.,  each component of $Z_*$ satisfies both \eqref{eq:propensity}-\eqref{eq:propensity1}, and  $V$ does not satisfy \eqref{eq:propensity1} and is a totally redundant covariate.
%Therefore, $V$ is the covariate vector excluded from both $p(Y|X)$ and $P(R=1|Y,X)$, which is defined as a redundancy of model.
To study (a)-(b), we develop a selection procedure to split $Z_*$ from $Z_S$ and eliminate some components of $X$ not in $X_*$ simultaneously.

%In this section, we further assume both $p(Y|X)$ and the instrument are unspecified. In our problem, we assume $X$ can be divided into $(U,Z^*,\varepsilon)$, where $(U,Z)$ is the covariates involved in the model and satisfying \eqref{eq:propensity}-\eqref{eq:propensity1}, and $\varepsilon$ is the covariates excluded from both $p(Y|X)$ and $P(R=1|Y,X)$. As we mentioned in Section 2.1, the validation criterion could only split $U$ from $X$, but could not distinguish $Z$ and $\varepsilon$. Hence, only when $p(Y|X)$ has a correctly specified parametric form, the selected $Z_{(\hat S)}$ is the most compact instrument. Otherwise, if $p(Y|X)$ is an overfitted model, $Z_{(\hat S)}$ is actually an estimation of $(Z,\varepsilon)$, which would lead to less efficient estimations for population parameters.

%By \eqref{eq:y given x} and the fact that $Z_S=(Z^*,V)$, it is easy to get that $V_A\cap V_B\cap Z_S\subset V$.

Since $
P(R=1|X)=\int p(y|X) P(R=1|y,U)dy$,
we obtain that
\begin{equation}\label{eq:y given xr}
p(Y|X,R=1)=\dfrac{p(Y|X)P(R=1|Y,U)}{\int p(y|X) P(R=1|y,U)dy},
\end{equation}
where $U$, given in
\eqref{eq:propensity}, contains all covariates not in $Z_S $.
If $X_A \subset X$ with the smallest dimension such that
\begin{equation}
Y \perp X \mid X_A, R=1 ,
%p(Y|X, R=1) = p(Y|X_A, R=1) ,
\label{eq:1func}
\end{equation}
i.e., $Y$ and $X$ are independent conditioned on $X_A$ and $R=1$,
then it follows from (\ref{eq:y given xr}) that $X_A \subset X_* \cup U$.
%It is desired to have $X_A = X_*\cup U$, which enables us to identify $X_*$ since $U$ is already identified.
%Note that $X_A$ can be find using  data  from $(Y,X)$ with $R=1$ and any nonparametric variable selection method.
However, if $p(Y|X)P(R=1|Y,U) = g(Y,X) h(X)$, where $g$ and $h$ are functions and $h(X)$ does not involve $Y$, then $h(X)$ is canceled from the ratio on right-hand side of (\ref{eq:y given xr}) so that $X_A$ misses $h(X)$ that could be part of $X_*$. To address this we can complement $X_A$ by finding $X_B$ with the smallest dimension such that
%\begin{align}
%p(Y|X)&=p(Y|X,R=1)P(R=1|X)+p(Y|X,R=0)P(R=0|X)\nonumber\\
%&=\frac{p(Y|X,R=1)P(R=1|X)}{P(R=1|Y,X)}\label{eq:y given x}.
%\end{align}
%From \eqref{eq:y given x}, $p(Y|X)$ involves $X$ through three functions, $P(R=1|Y,X)$, $p(Y|X,R=1)$ and $P(R=1|X)$.
%Note that the first function $P(R=1|Y,X)$ is the propensity, and it equals to $P(R=1|Y,U)$ by \eqref{eq:propensity}.
%It is shown in Section \ref{sec:ins} that we could find $U$ with probability tending to 1. Therefore, we only need to focus on the last two functions to do variable selection.
%Because of the existence of redundant covariate $V$, there exist two sub-vectors $X_A$ and $X_B$ of $X$ with the smallest dimensions, i.e., $X=(X_A,V_A)=(X_B,V_B)$ such that
\begin{equation}\label{eq:2func}
R \perp X \mid X_B
%P(R=1 |X) = P(R=1|X_B).
\end{equation}
Since $
P(R=1|X)=\int p(y|X) P(R=1|y,U)dy = \int g(y, X)h(X)dy $,  the part $h(X)$ will not be missed if it exists.
%Note that $X_B$ can be find using  data  from $(R,X)$  and any nonparametric variable selection method.  % and, hence, $X_* $ is contained in $(X_A \cup X_B) \cup U$.

From the previous discussion, it can be shown that
$Z_*=X_A\cup X_B \cap Z_S$ (see the proof of Theorem \ref{lem2} in the Appendix), i.e., the instrument vector
without any redundant component 	
	is a sub-vector of $Z_S$ containing components in either $X_A$ or $X_B$.

The next result can be used to simplify the two selection procedures (\ref{eq:1func})-(\ref{eq:2func}) into one selection procedure.

%, through variable selection (\ref{eq:1func}) and (\ref{eq:2func}), we can  address issue (a) raised in the beginning of this section, i.e., we can identify the redundant $V$ and the instrument $Z_*$ in $Z_S = (Z_*, V)$.

\begin{Theorem}\label{lem2}
	Let $X_A$ and $X_B$ be defined in   \eqref{eq:1func} and \eqref{eq:2func}, respectively, and $X_C$ be the smallest dimension covariate satisfying
	  \begin{equation}\label{eq:subc}
R \, e^Y \perp X \mid  X_C .
%p(Re^Y|X)= p(Re^Y |X_C).
\end{equation}
Then   $X_C=X_A\cup X_B$ and $Z_* = X_C \cap Z_S$.
\end{Theorem}

The transformation $e^Y$ is applied to avoid the confound of $Y=0$ with $R=0$.
Note that
the search for $X_C$ in (\ref{eq:subc}) can be carried out using data from $(R \, e^Y, X)$,
since $R \, e^Y =0$ when $R=0$ regardless of whether $Y$ is observed or not.
Thus, to implement variable selection with  (\ref{eq:subc}),
 we can apply an existing variable selection method using appropriate data.
Although $p(Y|X)$ is assumed to be parametric under (\ref{para}), the propensity is nonparametric and thus the distribution of $R \, e^Y$ conditioned on $X$ is nonparametric. Hence, we need to use a nonparametric variable selection method. Further, we would like to use a variable selection method that is selection consistent.
Let   $\hat X_C$ be the selected vector of covariates as an estimated  $X_C$. A variable selection method is selection consistent if
  $P(\hat X_C=X_C)\to 1$ as the sample size $N\to \infty$.
Variable selection  methods satisfy our needs include those in  \citet{bondell2009shrinkage}, \citet{huang2014feature}, and \citet{yu2016marginal}.
We propose to apply the shrinkage inverse regression method derived by \citet{bondell2009shrinkage}. The key idea of shrinkage inverse regression method is to employ Lasso shrinkage within the context of sufficient dimension reduction to achieve simultaneous dimension reduction and variable selection. Details of  the procedure and theoretical properties can be found in  \citet{ni2005note} and \citet{bondell2009shrinkage}.

 Let $\hat Z_S$ be obtained in Section 2 and $\hat U$ be the components of $X$ not in $\hat Z_S$.
 Then, we estimate $Z_*$  by
$\hat Z_*=\hat X_C\cap \hat Z_S$, and define
$\hat X_S =(\hat U, \hat Z_*)$.  The following result shows  asymptotic properties of $\hat Z_*$ and $\hat X_S$.

\begin{Theorem}\label{thm:all}
	Assume the conditions in Theorem 1 in \citet{chen2021instrument} and that
	$P(\hat X_C = X_C) \to 1$ as $N \to \infty$. \vspace{-2mm} Then,  as $N\to \infty$,\vspace{-2mm}
	\begin{enumerate}[(i)]
		\item  $P(\hat Z_*=Z_*)\to 1$; \vspace{-2mm}
		\item $P\{\hat X_S = (U, Z_*)\}\to 1$ and
		if, in addition, $U\subset X_*$, then $P(\hat X_S=X_*) \to 1$.
	\end{enumerate}
\end{Theorem}

Result (i) in Theorem \ref{thm:all} shows that $\hat Z_*$ is consistent, under the assumed conditions. Thus, issue (a) raised in the beginning of this section is well addressed. To address issue (b),
write $U=(U_A,U_B)$ with $U_A = X_*\cap U$  and $U_B \cap X_* = \varnothing$.
From the definitions,
\begin{equation}
	X_*=  (U_A ,Z_*)
\subset  (U,Z_*) =X_C \cup U . \label{xstar}
\end{equation}
Result (ii) in Theorem \ref{thm:all}  says that $\hat X_S$ is
 consistent in finding $(U,Z_*)$, but is asymptotically conservative in finding $X_*$
that is  strictly contained in $(U,Z_*)$ when  $U_B \neq \varnothing$. If
 $U_B = \varnothing$ ($U \subset X_*$), then $X_* = (U,Z_*)$ and $P(\hat X_S=X_*) \to 1$.
 In the presence of nonignorable nonresponse, we need the entire $U$ for a valid  inference, although part of $U$, $U_B$, is not in $X_*$.
 This can be viewed as  the price we pay for inference with nonignorable nonresponse.

A consistent selection of $U_A$ will result in a consistent selection of $X_*$, which
unfortunately cannot be done without any further condition on the propensity $P(R=1|Y,U)$ (e.g., finding $X_*$ can be done when nonresponse is ignorable), or on $p(Y|X)$. Note that the result in this section does not rely on the parametric assumption (\ref{para}). In the next section we show that, under some parametric model assumption on $p(Y|X)$, we can consistently find $X_*$.

\section{Parametric Model Selection}

The parametric model assumption (\ref{para}) is assumed for a valid inference with nonignorable nonresponse because \cite{robins1997toward} showed that
$p(Y|X)$ and $P(R=1|Y,X)$ cannot be both nonparametric
(assuming parametric $P(R=1|Y,X)$ is a different line of research as we stated in Section 1).
In applications, the parametric form in (\ref{para}) may have to be selected from a group of candidate models. For example, we may have a logistic model and a probit model for a binary $Y$, or we may take or not take logarithm for some continuous covariates and/or $Y$, or we may consider only linear functions of $X$ or both linear and quadratic functions of $X$.
 Let  $\mathcal{M}=\{ M_j :p_j(Y|X;\theta^{(j)}), j=1,\dots,J\}$ be a
 postulated $J$ candidate parametric models
 for (\ref{para}), where $\theta^{(j)}$ is an unknown parameter vector under model $M_j$. Based on observed data, we would like to select a model from $\mathcal{M}$, assuming at least one model in $\mathcal{M}$ satisfies (\ref{para}) instead of just a single model for (\ref{para}).

Parametric model selection regarding (\ref{para}) was studied in
 \citet{Fang2016}, but their method requires that a  known  instrument $Z$ be provided. Furthermore, \citet{Fang2016} proposed to consider all sub-vectors of $X$ which may be computationally infeasible.
	
Because of the variable selection procedure in Section 3, we can assume that each $M_j$ in $\mathcal{M}$ contains all components of $X$. For example, if $X=(X_1,X_2)$ and $Y|X \sim N(\beta_0 + \beta_1 X_1+ \beta_2 X_2, \sigma^2)$ is in $\mathcal{M}$, then  $Y|X \sim N(\beta_0 + \beta_1 X_1, \sigma^2)$ should not be in $\mathcal{M}$, but $\mathcal{M}$ may contain $Y|X \sim N(\beta_0 + \beta_1 X_1+ \beta_2 X_2+ \gamma_1 X_1^2 + \gamma_2 X_2^2, \sigma^2)$ or
$\log Y | X \sim N(\beta_0 + \beta_1 X_1+ \beta_2 X_2, \sigma^2)$.

Let $\mathcal{M}= \mathcal{M}_C \cup \mathcal{M}_W$, where  $\mathcal{M}_C$ contains all correct models and $\mathcal{M}_W$ contains all wrong models.
Note that models in $\mathcal{M}_C$ are nested. A correct model may contain redundant covariates or redundant components such as an unnecessary high order term of a useful covariate. Our goal is to find the most compact correct model, which is called the best model.

Using the methods proposed and studied in Sections 2-3 as steps 1-2, we now add two more steps, one for consistent model selection to decide which $M_j$ should be adopted, and one for consistently finding $X_*$, which is not resolved as we discussed in the end of Section 3.
The following is a summary of the entire 4-step procedure, including some discussions and explanations at each step.
\begin{enumerate}[Step 1.]
		\item First, we apply the validation criterion proposed in Section \ref{sec:ins} on each model $M_j\in  \mathcal{M}$ to select a sub-vector of $X$ satisfying \eqref{eq:propensity}. Denote this sub-vector of $X$ as $\hat Z_S^{(j)}$, $j=1,...,J$. If $M_j \in \mathcal{M}_C$, then the theoretical result in Section 2 ensures that,
	with probability tending to 1 as $N \to \infty$, $\hat Z_S^{(j)} = Z_S = (Z_*,V)$,
		 where each component of $Z_*$ satisfies both \eqref{eq:propensity}-\eqref{eq:propensity1} and  $V$ satisfies \eqref{eq:propensity} but not \eqref{eq:propensity1}, i.e.,  $V$ is a totally redundant covariate vector.
%		 	If  $M_j \in \mathcal{M}_W$, 	 then the theoretical result in Section 2 ensures that,  with probability tending to 1 as $N \to \infty$,
%$\hat Z_S^{(j)} = V$ or $\hat Z_S^{(j)} = X$ (when VC$(l_1)$ in (\ref{eq:d}) does not tend to 0).
		\vspace{-2mm}
		\item Second, we apply the nonparametric variable selection method proposed in Section \ref{sec:var} to find $\hat X_C$. Note that this part does not depend on model $M_j$.
		Let $\hat Z_*^{(j)}=\hat X_C\cap \hat Z_S^{(j)}$, $j=1,...,J$. The theoretical result in Section 3 shows that, with probability tending to 1 as $N \to \infty$, $\hat Z_*^{(j)} = Z_*$
		if $M_j \in \mathcal{M}_C$.
%		and $\hat Z_*^{(j)} = \varnothing$ or $\hat X_C$ if  $M_j \in \mathcal{M}_W$.
		%After this step, we have found the instrument with the smallest dimension and identified the correct models in $\mathcal{M}$.
		\vspace{-2mm}
		\item
		Third, we carry out model selection with models in $\mathcal{M}$.
		Let $\hat U^{(j)}$ be the vector containing components of $X$ that are not in $\hat Z_S^{(j)}$, $j=1,...,J$. Since in step 2 we reduce $X$ to $\hat X_S^{(j)} = \hat X_C \cup \hat U^{(j)}$ under model $M_j$, we can reduce model $M_j$ to
		$\tilde M_j : p_j(Y| \hat X_S^{(j)}; \theta^{(j)})$.
	We proposed to apply the penalized validation criterion  (PVC)  in \citet{Fang2016} to select the most compact model from  	
	$\{ \tilde{M}_j:  j=1,...,J \}$, except that we use $\hat Z_*^{(j)}$ selected in Step 2 as the  instrument under $j$th model instead of a fixed given instrument for all models in
	\citet{Fang2016}.
Theoretically, we can establish a result similar to that in \citet{Fang2016}, i.e.,
with probability tending to 1 as $N \to \infty$, we will select the most compact correct model in 	\vspace{-2mm} $\mathcal{M}$.
\item
Finally, we would like to finish with finding the most compact covariate $X_*$ defined by (\ref{bestcov}). This not only completes the variable selection started at Step 2, but also leads to the most compact model for the final estimation of the parameter $\theta$ using the pseudo likelihood 	in (\ref{likelihood}). Suppose that the model selected in Step 3 is the one with $j=j_*$. Then $\hat Z_* = \hat Z_*^{(j_*)}$ is the selected instrument vector,
$\hat U_* = \hat U^{(j_*)}$ is the selected covariate vector in the propensity, and
$\hat X_S^{(j_*)} = \hat X_C \cup \hat U_*$ is the selected covariate vector in the selected model $j_*$.
According to (\ref{xstar}), we just need to split   $\hat U_*$ into $\hat U_A$ and $\hat U_B$, where $\hat U_A$ is included in $p(Y|X)$ and $\hat U_B$ is excluded.
To achieve this, we apply the PVC in \citet{Fang2016} again to select  a
model from $\{ p_t (Y|\hat X_C \cup \hat U_t ; \theta^{(t)} ): \hat U_t \mbox{ is a sub-vector of }\hat U_*\}$, using $\hat Z_*$ as the instrument. If $\hat U_{t_*}$ is selected, then our selected $X_*$ in (\ref{bestcov}) is $\hat X_* = (\hat U_{t_*} , \hat Z_*)$. The final model selected is $\hat M_* : p_{t_*}(Y|\hat U_{t_*}, \hat Z_* ; \theta^{(t_*)} )$.

%Note that $P\{ \hat X_S^{(j_*)} = (X_C\cup U)\} \to 1$, where $X_C$ is given in (\ref{eq:subc}) and $U$ is the most compact covariate in the propensity.
%	Let $\hat{\mathcal{M}}_C = \{ M_j: \hat Z_*^{(j)} \neq \varnothing \}$, which is the correct models identified in step 2.
	%	Theoretically, $P (\hat{\mathcal{M}}_C = \mathcal{M}_C) \to 1$ as $N \to \infty$ and, hence, we just need to choose the model in $\hat{\mathcal{M}}_C$ with the smallest dimension, since all models in $\mathcal{M}_C$ are nested.
%		However, in applications the sample size is finite and $\hat{\mathcal{M}}_C$ may include a few wrong models.
	%	To enhance the performance of our selection procedure,
\end{enumerate}	

 Once the instrument $\hat Z_*$, covariate vectors $\hat U_*$ and $\hat X_*$, and model $\hat M_*$ are selected, we can maximize the pseudo likelihood in (\ref{likelihood}) to estimate the corresponding parameter $\theta$ in  model $\hat M_*$ and other parameters such as the mean of $Y$. The following result summarizes asymptotic properties of estimators under our approach.

 \begin{Corollary}\label{thm:last}
 	Assume the conditions in Theorem \ref{thm:all} and condition C8 in
 	\citet{Fang2016}.  \vspace{-2mm} Then,  as $N\to \infty$,
 	\begin{enumerate}[(i)]
 		\item $P(\hat X_* =X_*)\to 1$; \vspace{-2mm}
 		\item $P(\hat M_* = \mbox{the most compact correct model }) \to 1$; \vspace{-2mm}
 		\item  the estimator of $\theta$ obtained by maximizing the pseudo likelihood in (\ref{likelihood}) with selected $\hat Z_*$, $\hat U_*$, $\hat X_*$, and  model $\hat M_*$ is consistent and asymptotically normal.
 	\end{enumerate}
 \end{Corollary}

\section{Simulation Studies}
In this section, we study the finite-sample performance of the proposed method in terms of the rate we select  instrument $Z_*$, covariate vector $X_*$, and the most compact correct model. All the results are based on 1000 simulation replications and three sample sizes $N=100,\;200,\;500$.
\subsection{Simulation Study 1: Instrument Selection}
In the first simulation study, we consider the situation in which there is no redundant covariate in $p(Y|X)$, i.e., $X=X_*$, and the instrument is unknown. Thus, this simulation evaluates the performance of the procedure in Section 2 for instrument selection only.

Let $X=(X_1,X_2,X_3)$, where $X_j \sim \chi^2(1),\; j=1,2,3$, are independent, $Y|X \sim N(\beta_0+\beta_1 X_1+\beta_2 X_2+\beta_3 X_3, \sigma^2)$, where $\theta=(\beta_0,\beta_1,\beta_2,\beta_3,\sigma^2)=(5,2,2,2,3)$.
 For the nonresponse propensity, we consider the following three cases:
$$
\begin{array}{llll}
\text{Case 1:} & P(R=1|Y,X) \!\! \! & =\left[1+\exp\left\lbrace  -Y+8  \right\rbrace\right]^{-1}, & U= \varnothing , Z=X;\\
\text{Case 2:} & & = \left[1+\exp\left\lbrace  -Y(4X_1-1)  \right\rbrace\right]^{-1},
& U= X_1 , Z=(X_1,X_2);\\
\text{Case 3: }& &=\left[1+\exp\left\lbrace  -1.2Y(8X_1X_2-1) \right\rbrace\right]^{-1},
& U=(X_1,X_2), Z=X_3 .
\end{array}
$$
The unconditional response rates for these three cases are 68\%, 63\% and 52\%, respectively. Because all $\beta_j$, $j=1,2,3$, are nonzero, there is no redundant covariate and $Z_* = Z_S =  Z$.
Since we assume the existence of instrument, there are $2^3-1=7$ candidate instruments.

Table \ref{simu1:selection rate} reports the number of times of selecting each candidate instrument by the proposed validation criterion (\ref{eq:d}). It can be seen that the proposed method can select the instrument $Z_*$ with empirical probability much higher than those for other candidates.  When the sample size is 200 or larger, the probability of correctly selecting $Z_*$ equal or close to 1.

Besides instrument selection, we also consider the estimation of $\theta$ and $E(Y)$ using different instruments and the pseudo likelihood (\ref{likelihood}).
%Nine methods are compared: full data analysis assuming no missing data (FULL), pseudolikelihood estimation (PL) method with seven candidate instruments and the selected instrument $\hat Z_S$.
Table \ref{simu1:estimation} reports the empirical means and standard derivations of estimators in case 2 and $N=200$, based on 1000 simulations.
With missing data,  the only data-adaptive estimators are those based on $\hat Z_S = \hat Z_*$, and the rest estimators based on a fixed choice of instrument are entered for comparison.
 In this setting,
the best instrument is $Z_* = (X_2,X_3)$, but both $X_2$ and $X_3$ are correct one-dimensional instruments satisfying \eqref{eq:propensity}-\eqref{eq:propensity1}.
The top part of Table \ref{simu1:estimation} shows the results when a correct instrument  or our proposed $\hat Z_*$ is used; it can be seen that estimators are almost unbiased but the use of $Z_*$ produces the most efficient estimators, and estimators based on the proposed $\hat Z_*$ are almost the same as those based on $Z_*$.
On the other hand, when one of $X_1$, $(X_1,X_2)$, $(X_1,X_3)$, and $(X_1,X_2,X_3)$ is used as an instrument, some estimators are seriously biased because of using a wrong instrument.
The last row of Table \ref{simu1:estimation} reports the results for estimators when there is no missing data, which provides a standard of the best we can do.

\subsection{Simulation Study 2: Instrument, Variable, and Model Selection}
In the second simulation study, we consider the situation where $p(Y|X)$ involves some redundant covariates and we also need to choose a model from three models for (\ref{para}). Hence, this simulation evaluates the performance of our four-step selection procedure proposed in Section 4.

We consider a 10-dimensional covariate $X=(X_1,\dots,X_{10})$, where $(X_1,X_2,X_3)$ and $Y$ are generated the same as those in simulation study 1 in Section 5.1, and $X_4,\dots, X_7$ are generated independently from the standard normal distribution and are redundant, i.e., $p(Y|X)=p(X_1,X_2,X_3)$.
In this experiment, we consider the following three models for (\ref{para}):
\begin{align*}
M_1: Y|X &\sim N\left(\beta_0+\beta_1 X_1 + \cdots + \beta_{10}X_{10}, \, \sigma^2\right),\\
M_2: Y|X &\sim N\left(\beta_0+\gamma_1 X^2_1 + \cdots + \gamma_{10}X^2_{10}, \, \sigma^2\right),\\
M_3: Y|X &\sim N\left(\beta_0+\beta_1 X_1 + \cdots + \beta_{10}X_{10} +\gamma_1 X^2_1 + \cdots + \gamma_{10}X^2_{10},
\, \sigma^2\right),
\end{align*}
where  $M_1$  is a correct model,  $M_2$ is a wrong model, and
$M_3$ is correct but overfitted.
The best model, i.e., the most compact correct model, is  $M_1$ with $\beta_4 = \cdots = \beta_{10}=0$. As we discussed in Section 4, we do not need to consider sub-models of each $M_j$, $j=1,2,3$, because of the variable selection step described in Section 3.
%For example, if $M_1$ is selected in step 3, then, together with the variable selection in step 2, model $M_1$ with $\beta_4 = \cdots = \beta_{10}=0$ is selected.

For the propensity, we consider the three cases in Section 5.1, and  an  additional case,
\[
\text{Case 4}:\;P(R=1|Y,X)=\left[1+\exp\left\lbrace  -Y(X_4+0.5) \right\rbrace\right]^{-1}, \quad U=X_4, \, Z= (X_1,X_2,X_3).
\]
The unconditional response rate for case 4 is 69\%.
The reason we consider the additional case 4 is that it is the case where $p(Y|X)$ is a function of $X_*=(X_1,X_2,X_3)$ and $U= X_4$ is not in $X_*$, whereas in all cases 1-3, $U \subset X_* =(X_1,X_2,X_3)$.
 Therefore, case 4 allows us to evaluate step 4 of our procedure in finding the most compact covariate $X_* =
 (X_1,X_2,X_3)$ when $U \not \subset X_*$.

Table \ref{simu2} reports the frequencies in 1000 simulations of correctly selecting  $U$, $Z_*$,  $M_1$, and $X_*$  using the proposed  four-step procedure in Section 4.
It can be seen that, in cases 1-2 and 4, all correct selection probabilities are higher than 90\% and some are close to 1 when $N =200$ and 500, and are between 70\% and 90\% when $N=100$. Case 3 is the most difficult situation for instrument search since $Z_*$ is one-dimensional and $X$ is 10-dimensional; the correct selection probabilities are too low when $N=100$, but are adequate when $N=200$, and are close to 1 when $N=500$.

%Although	the rates of selecting $X_1$ and $X_2$ as $U$ are greater than 80\%, the correct selection rate is less than 40\%,

While Table \ref{simu1:estimation} for simulation study 1 shows some gain of using a 2-dimensional $Z_*$ over using only a part of $Z_*$ in the estimation of $\theta$, we include estimation results in Table \ref{simu2:estimation} for simulation study 2 to illustrate that the gain of using $\hat Z_*$ as instrument over using $\hat Z_S$ that contains $\hat Z_*$ and some redundant covariates.

\section{Real Data Analysis}

\subsection{National Health and Nutrition Examination Survey}
To illustrate  our proposed instrument selection method, we consider a real data set from the National Health and Nutrition Examination Survey (NHANES) conducted in 2005 by the United States Centers for Disease Control and Prevention, which was also analyzed in \citet{Fang2016} from model selection perspective. In this data set, $Y$ is the body fat percentage, which is measured by dualenergy
x-ray absorptiometry and denoted as dxa, and the covariates are age, gender, and body mass index (bmi), i.e., $X=(\text{bmi, age, gender})$. As in  \citet{Fang2016},
we consider  $N=1591$ middle-aged and elderly people, from whom 393 (24.7\%) have missing dxa.

 \citet{Fang2016} assumed that $(\text{age,gender})$ is a given (known) instrument,
based on researcher's experience, and considered model selection with the following four candidate models:
	\begin{align*}
M_1:&\; Y|X \sim N(\beta_0+\beta_1 \text{bmi}+\beta_2 \text{age} + \beta_3 \text{gender},\sigma^2 );\\
M_2:&\; Y|X \sim N(\beta_0+\beta_1 \text{bmi}+\beta_2 \text{age} + \beta_3 \text{gender}+\beta_4 \text{age $\times$ gender},\sigma^2 );\\
M_3:&\; \log Y | X\sim N(\beta_0+\beta_1 \log(\text{bmi})+\beta_2 \text{age} + \beta_3 \text{gender},\sigma^2 );\\
M_4:&\; \log Y|X \sim N(\beta_0+\beta_1 \log(\text{bmi})+\beta_2 \text{age} + \beta_3 \text{gender}+\beta_4 \text{age $\times$ gender},\sigma^2 )
\end{align*}

 We now apply and illustrate our four-step procedure in Section 4, which starts with a search for an instrument.
 %First, we apply (\ref{eq:d}) to obtain $\hat d$.
 Since gender is a binary covariate, as we discussed after procedure (\ref{eq:d}) in Section 2, we need to combine gender with either age or bmi before applying (\ref{eq:d}). Then, we have the following $\tilde p= 2p-2 = 4$ candidates: $\tilde Z_1 = \text{age}$, $\tilde Z_2 = \text{bmi}$, $\tilde Z_3 = (\text{age, gender})$, and $\tilde Z_4 = (\text{bmi, gender})$.

Table \ref{data:NHANES} (Step 1) reports  the $\vc$ value in (\ref{eq:VC}) for each candidate instrument.
Based on \eqref{eq:d} and $\hat Z_S $ from Step 1, Table \ref{data:NHANES} (Step 2) shows $\hat{Z}_*$ obtained under each model $M_j$.
Table \ref{data:NHANES} (Step 3) lists the values of PVC in
 \citet{Fang2016} using $\hat Z_*$ as the instrument.
The smallest PVC$(\hat Z_*)$ corresponds to model $M_1$. Thus, $M_1$ is selected with instrument $\hat Z_*$ = age and $\hat U$ = (bmi, gender).

While $M_1$,  $\hat Z_*$ = age, and $\hat U$= (bmi,gender) are selected after Steps 1-3, we still need to find out whether some or all components of $\hat U$ = (bmi, gender) should be in $M_1$, i.e., what $X_*$ is. Table \ref{data:NHANES} (Step 4)
shows the values of PVC in
\citet{Fang2016}  under four sub-models of $M_1$ with different $X_*$ ($\hat Z_*$ = age must be in $X_*$).  The smallest PVC leads to $\hat X_*$ = (bmi,age,gender) = $X$. The final fitted regression model is given in the last part of Table \ref{data:NHANES}.

\subsection{AIDS Clinical Trials Group Protocol 175 Dataset}
As another example, we apply our variable and instrument selection procedure to the data collected on 2139 HIV positive patients enrolled in AIDS Clinical Trials Group Protocol 175 (ACTG175).
Patients in this HIV clinical trial were randomized into four arms to receive the respective antiretroviral regimen: (1) zidovudine (ZDV) with 532 subjects; (2) didanosine (ddi) with 522 subjects; (3) ZDV + ddi with 524 subjects; and (4) ZDV + zalcitabine with 561 subjects. The response $Y$ is the CD4 cell count at $96\pm 5$ weeks from patients receiving  antiretroviral regimen. The covariate vector $X$ contains 6 components,
$X_1 = \text{age}$, $X_2 = \text{weight}$, $X_3 = \text{CD4}$ cell count at baseline, $X_4 = \text{CD4}$ cell count at $20 \pm 5$ weeks, $X_5 = \text{CD8}$ cell count at baseline, and $X_6 = \text{CD8}$ cell count at $20 \pm 5$ weeks. The purpose of this study is to examine whether the new regimens (2), (3) and (4) work better than the
traditional regimen (1) in the sense that the mean of $Y$ is higher.
We analyze data separately in four arms to estimate $E(Y)$ under each regimen.

Due to  death and dropout, $Y$ has missing values, but $X$ is always fully observed. The missing rate of $Y$ are 39.66\%, 36.21\%, 35.69\% and 37.43\%, respectively, in four arms. We assume the following model in (\ref{para}) for each arm,
\[
Y | X \sim N(\beta_0+\beta_1 X_1 +\beta_2 X_2+\beta_3 X_3+\beta_4 X_4+\beta_5 X_5+\beta_6 X_6,\sigma^2).
\]
To use the pseudo likelihood (\ref{likelihood}),  we also assume that $X$ follows a multivariate normal distribution whose parameters are estimated by maximum likelihood estimates.

Table \ref{data:ACTG} reports the results for four arms. In each arm, we first calculate the VC value (\ref{eq:VC})  under each $X_k$ as a candidate of instrument. For all  four arms, based on \eqref{eq:d}, we obtain that $\hat Z_S=X$ and $\hat U=\varnothing$.
After applying Step 2 of our procedure, we find that $\hat Z_* = (X_3,X_4)$ in arms 1-2 and 4, and  $\hat Z_* = (X_2,X_3,X_4,X_6)$ in arm 3. This means that $X_1 =$ age and $X_5 = $ CD8
cell count at baseline are redundant in all arms, and
$X_2 = $ weight and $X_6 =$ CD8 cell count at $20 \pm 5$ weeks are also redundant in arms 1-2 and 4. Because $X_3 =$ CD4
cell count at baseline and $X_4 =$ CD4 cell count at $20 \pm 5$ weeks are good predictors of $Y =$  CD4 cell count at $96 \pm 5$ weeks, once $X_3$ and $X_4$ are present, other covariates become unimportant or nearly unimportant.
$(X_3,X_4)$ is also an instrument in four arms, and only in arm 3,
a better instrument is $(X_2,X_3,X_4,X_6)$. Since $\hat U = \varnothing$, $\hat X_* = \hat Z_*$, and we do not need Step 4 of our procedure.
The reason why $\hat U = \varnothing$ might be that the latest CD4 cell counts $Y$ is the dominating variable for the propensity.

After instrument and variable selection, we report the estimates and standard deviations of regression coefficients for the components of $X_*$  and $E(Y)$ for each regimen in Table \ref{data:ACTG}.
In all  arms, the estimates for $E(Y)$ by the proposed method are smaller than the corresponding estimates of $E(Y)$ computed based on the assumption that $Y$ values are missing at random. This indicates that the nonresponse of $Y$ may be nonignorable.
	Finally,  the estimates of $E(Y)$ in arms 2-4  are greater than that in arm 1, the traditional regimen, which suggests that the new regimens indeed work better than the traditional regimen.

\section{Conclusion}

In this paper, we have developed a comprehensive methodology for simultaneous instrument search, variable selection, and model specification in the presence of nonignorable nonresponse. By introducing a pseudo-likelihood-based validation criterion, our approach effectively identifies a valid instrument from a set of covariates, a crucial step for ensuring identifiability and obtaining consistent estimators. The proposed method further integrates dimension reduction and parametric model selection, leading to a coherent framework for valid statistical inference. Simulation studies and real-data applications confirm that the procedure performs well in finite samples.

The framework for instrument search established here is general and can be potentially extended to a broader class of statistical models with nonignorable missingness, such as estimating equations \citep{chen2018semiparametric, chen2019semiparametric}, longitudinal or functional regression model \citep{tseng2016longitudinal, chen2018functional, lin2023flexible}, or survival model \citep{Zhang_Ling_Zhang_2024}. The core principle of using a valid instrument offers a versatile framework for tackling nonignorable nonresponse in a wide range of statistical problems.

Beyond its theoretical appeal, the instrument approach provides a practical and robust solution for enhancing the accuracy of statistical estimation and inference in real-world applications involving nonignorable nonresponse. In clinical trials with informative dropout \citep{little2012prevention, ratitch2013missing}, social reporting systems with incomplete data collection \citep{faust2025firearm, nakatsuka2025all}, and registry data with incomplete follow-up \citep{wei2018u, chen2021natural, gottlieb2025work, gottlieb2025differences, openshaw2025effect}, our data-driven instrument selection framework enables researchers to obtain more reliable parameter estimates and valid conclusions. By addressing the critical challenge of identifiability through automated instrument discovery, this approach offers a principled way to handle nonignorable missing data across various scientific domains.

\section*{Appendix: Proofs}

\begin{proof}[Proof of Theorem \ref{lem2}.]
	Let $X=(X_A,V_A)=(X_B,V_B)$, where $X_A$ and $X_B$ are defined in \eqref{eq:1func} and \eqref{eq:2func}. Thus, $V_A$ and $V_B$ are the covariates not involved in $p(Y|X,R=1)$ and $P(R=1|X)$, respectively.
	From \eqref{eq:y given xr} and
	\begin{equation}\label{eq:r given x}
	P(R=1|X)=\int p(y|X)P(R=1|y,U)dy,
	\end{equation}
	$P(R=1|X)$ and $P(Y|X,R=1)$ both involve $X$  through $p(Y|X)$ and $P(R=1|Y,U)$. Since the redundant $V$ could be excluded from both $p(Y|X)$ and $P(R=1|Y,U)$, functions on the right side of \eqref{eq:y given xr} and \eqref{eq:r given x} do not involve $V$. Therefore, $V\subset V_A$ and $V\subset V_B$. Since we also have   $V\subset Z_S$, we conclude that $V\subset V_A\cap V_B\cap Z_S$.	
	Next, \eqref{eq:y given xr} could be re-written as
	\begin{equation}\label{eq:yx}
	p(Y|X)=\dfrac{p(Y|X,R=1)P(R=1|X)}{P(R=1|Y,U)}.
	\end{equation}
	Note that $U$ is the components of $X$ not in $Z_S$. By \eqref{eq:yx}, $V_A\cap V_B\cap Z_S\subset V$.
	Combining the above arguments, we obtain
	\begin{equation}\label{eq:v}
	V=V_A\cap V_B\cap Z_S.
	\end{equation}
	Note that, $Z_*$ contains the components of $Z_S$ not in $V$, and the components of $Z_S$ not in the right side of \eqref{eq:v} form $X_A\cup X_B\cap Z_S$. Therefore, $Z^*=X_A\cup X_B\cap Z_S$.

It remains to prove $X_C = X_A \cup X_B$ with $X_C$ defined in (\ref{eq:subc}). Note that
\begin{align}
P(R\, e^Y=0|X)& =P(R=0|X),\nonumber\\
P(R\, e^Y\leq t |X)&=P(R=1,Y\leq \log t |X)+P(R=0|X)\label{eq:cab}\\
&=P(Y\leq \log t |R=1,X)P(R=1|X)+P(R=0|X)\nonumber
\end{align}
for any $t > 0$.
It follows from (\ref{eq:cab}) and the definition of $X_C$ that $X_C\subset X_A\cup X_B$. On the other hand,
if \eqref{eq:subc} holds, then $P(R\, e^Y=0|X)=P(R\, e^Y=0|X_C)$, which together with the first equation in (\ref{eq:cab}) imply that $P(R=0|X)=P(R=0|X_C)$, i.e., $X_B\subset X_C$. This, coupled with \eqref{eq:cab} and the fact that \eqref{eq:cab} holds with $X$ replaced by $X_C$, imply that $P(Y\leq y|R=1,X)=P(Y\leq y|R=1,X_C)$ for any $y \in (-\infty , \infty )$. Therefore, $X_A\subset X_C$ and $X_A\cup X_B\subset X_C$. This completes the proof.
\end{proof}

\begin{proof}[Proof of Theorem \ref{thm:all}.]
The result follows directly from Theorems in \citet{chen2021instrument} and Theorems \ref{lem2}.
\end{proof}

\bibliographystyle{apalike}
\bibliography{InstrumentSelection}

\newpage

\begin{table}[tbp]
	\centering
	\caption{Number of times validation criterion (\ref{eq:d}) selects each instrument in 1000 replications in simulation study 1\label{simu1:selection rate}}\vspace{2mm}
	\begin{tabular}{ccccccccccc}
		\hline
		&		&& &\multicolumn{7}{c}{Selected instrument by (\ref{eq:d})}\\
		Case &$Z_*$ & 	$N$&&$X_1$ & $X_2$ & $X_3$ &  $(X_1,X_2)$ & $(X_1,X_3)$ & $(X_2,X_3)$& $(X_1,X_2,X_3)$ \\ \cline{5-11}
			1&$(X_1,X_2,X_3)$ &   	100 &&  4  &   9  &   7  &   3  &   4 &    8  & 965 \\
	&& 200 &   &    6  &   6 &    4 &    7&     2  &   0 &  975 \\
&&   	500& &2  &  2  &  1  &  3  &  2  &  4  &  986 \\
	2&$(X_2,X_3)$  &    100 &&        0 &   13  &  16  &   0  &   0 &  970   &  1 \\
		&		 & 200&   &   0 &    9  &   6  &   0  &   0  & 985 &    0 \\
		&		&500   &   &0  &  0  &  0  &  0  &  0  &  $\!\! 1000  $&  0\\
 3&$X_3$ &  100 &    &4  &   1  & 707  &   0 &   54 &   55 &  179\\
		&		&200 &      &0  &   0  & 949  &   0   &  8  &   5  &  38 \\
		&	    &500 &   &0  &  0  & $\!\! 1000$  &  0  &  0  &  0  &  0  \\
		\hline
	\end{tabular}
\end{table}

\begin{table}[htbp]
	\centering
	\caption{Mean and standard deviation (in parentheses) of estimators using different $Z$'s  in Case 2, $N=200$, based on 1000 replications in simulation study 1\label{simu1:estimation}} \vspace{2mm}
	\begin{tabular}{ccccccc}
		\hline
		& \multicolumn{6}{c}{Parameter} \\
		$Z$ used	& $\beta_0=5$&   $\beta_1=2$  &   $\beta_2=2$  &  $\beta_3=2$  &  $\sigma^2=3$ & $E(Y)=11$  \\
		\cline{2-7}
		$X_2$               &             4.99 (0.43)  &   2.01 (0.16) &    2.01 (0.16)   &   2.02 (0.19)  &    2.91 (0.52)   &  11.01 (0.41)         \\
		$X_3$              &           5.00 (0.43)   &   2.00 (0.16)  &    2.02 (0.18)   &   2.01 (0.16)  &    2.92 (0.54)  &  11.01 (0.42)         \\
		$(X_2,X_3)$          &            5.00 (0.35)   &   2.00 (0.13) &    2.00 (0.13)   &   2.00 (0.13)  &    2.91 (0.44)   &  10.99 (0.39)       \\
		$\hat Z_S = \hat Z_*$      &            5.00 (0.35)   &   2.00 (0.13)  &    2.00 (0.14)   &   2.00 (0.13)  &    2.91 (0.44)   &  10.99 (0.39)         \\
		$X_1$                &             3.35 (0.40)   &   2.50 (0.17)  &    2.00 (0.14)  &   2.00 (0.14)  &    3.82 (0.63)   &   \ 9.82 (0.48)         \\
		$(X_1,X_2)$           &            4.10 (0.34)   &   2.23 (0.13)  &    2.13 (0.14)   &   2.00 (0.14)  &    3.29 (0.48)   &  10.44 (0.44)          \\
		$(X_1,X_3)$           &            4.10 (0.34)  &   2.23 (0.13)  &    2.00 (0.13)   &   2.13 (0.14)  &    3.29 (0.48)  &  10.44 (0.44)         \\
		$(X_1,X_2,X_3)$     &            4.41 (0.32)   &   2.13 (0.12)  &    2.07 (0.13)   &   2.07 (0.13)  &    3.11 (0.43)   &  10.66 (0.42)
		\\
		No missing              &             4.99 (0.19)   &   2.00 (0.09)  &    2.00 (0.09)  &   2.00 (0.09) &    2.93 (0.30)  &  10.98 (0.36)               \\
		\hline	
	\end{tabular}
\end{table}

\begin{table}[htbp]
	\centering
	\caption{Frequency of correct selection in 1000 replications at each step in simulation study 2\label{simu2}}\vspace{2mm}
	\begin{tabular}{cccccccccc}
		\hline
%		& & \multicolumn{2}{c}{$X_*= (X_1,X_2,X_3)$} \\
	Case & $U$ &$Z_*$ & $X_*$ & $N$	&&Step 1 & Step 2 & Step 3 & Step 4 \\
 	\cline{7-10}
	1 & $\varnothing$ & $(X_1,X_2,X_3) $ &$(X_1,X_2,X_3)$& 100	&&	913   & 892 & 884  & 884\\
&&&& 200 &&	928 & 912 & 912 & 912 \\
&&&&	500 &&939    & 927    &927 & 927 \\
2&  $X_1$ &$(X_2,X_3)$ & $(X_1,X_2,X_3)$&100 && 959   & 938  &855  &   845 \\
&&&& 200 &&  982 & 975   & 936  &  936  \\
&&&&500 & & $ \!\!  1000 $   & 998    & 998   &   998        \\
3&  $(X_1,X_2)$&$X_3$ & $(X_1,X_2,X_3)$& 100& & 392 &369   & 317  &   272   \\
&&&& 200 &   &868 & 848  & 789  &  767    \\
&&&&500&  & 998 &  993     & 985  &     985      \\
	4& $X_4$ &$(X_1,X_2,X_3)$ & $(X_1,X_2,X_3)$& 100 &&832 & 785   & 714  &    713       \\
&&&& 200 && 985 & 973  & 952  &  952 \\
&&&&500 &    &999    &  992    & 992 &    992   \\
		\hline
\multicolumn{8}{l}{	{\small 	Step 1: correctly selecting $U$ under model $M_1$;}} \vspace{-2mm} \\
		\multicolumn{8}{l}{\small Step 2: correctly selecting $U$ and $Z_*$ under model $M_1$;} \vspace{-2mm}\\
		\multicolumn{8}{l}{\small Step 3: correctly selecting $U$, $Z_*$ and $M_1$;} \vspace{-2mm}\\
		\multicolumn{8}{l}{\small Step 4: correctly selecting $U$, $Z_*$, $M_1$ and $X_*$. } \\
	\end{tabular}
\end{table}

\begin{table}[htbp]
	\centering
	\caption{Mean and standard deviation (in parentheses) of estimators using different $Z$'s, based on $N=200$ and 1000 replications in simulation study 2\label{simu2:estimation}} \vspace{2mm}
	\begin{tabular}{cccccccc}
		\hline
	&	& \multicolumn{6}{c}{Parameter} \\
Case &	$Z$ used	& $\beta_0=5$&   $\beta_1=2$  &   $\beta_2=2$  &  $\beta_3=2$  &  $\sigma^2=3$ & $E(Y)=11$  \\
\cline{3-8}
1 &		$\hat Z_S$ &    5.07 (0.35) &   1.99 (0.13) &    1.99 (0.14) &    1.99 (0.14) &    2.76 (0.42) &   11.03 (0.40)\\
&		$\hat Z_*$ &    5.04 (0.32) &    1.99 (0.10) &    1.99 (0.10) &    1.99 (0.11)&    2.90 (0.40) &   11.01 (0.38) \\
2 &		$\hat Z_S$ &        5.00  (0.38)&  2.01 (0.16) &  2.00 (0.15) &  2.01 (0.14) &  2.75 (0.45)&  11.00 (0.37)\\
&		$\hat Z_*$ & 5.00  (0.34)&  2.01 (0.13) &  2.00 (0.13) &  2.01 (0.12)&   2.90 (0.44) &  11.00 (0.35)\\
3 &		$\hat Z_S$ &      4.95 (0.55) &    2.02 (0.22) &    2.02 (0.21) &    2.02 (0.22) &    2.73 (0.65) &   11.00 (0.45)  \\
&		$\hat Z_*$ &    4.98 (0.51) &  2.01 (0.18) &   2.01 (0.17)  &  2.02 (0.18) &   2.89 (0.60) &  11.01 (0.42) \\
4 &		$\hat Z_S$ &       4.98 (0.37) &   2.00 (0.14) &  2.00 (0.15) &    2.00 (0.14) &    2.78 (0.43) &   10.99 (0.42) \\
&		$\hat Z_*$ &    5.00 (0.29)  &  2.00 (0.12)  &  2.00 (0.12) &    2.00 (0.12)  &    2.91 (0.39)  &   11.00 (0.37) \\ \hline
	\end{tabular}
\end{table}

\begin{table}
	\centering
	\caption{Instrument, variable, and model selection results in NHANES dataset\label{data:NHANES}} \vspace{2mm}
		\begin{tabular}{ccccccc}
			\hline
			& \multicolumn{4}{c}{Step 1: $\vc$ in (\ref{eq:VC})}& Step 2 & Step 3\\
			Model 	& bmi & age & (age, gender) & (bmi, gender) & \underline{$~~~ \hat Z_* ~~~$} & \underline{$\text{PVC}(\hat Z_*)$}\\
			\cline{2-5}
			$M_1$ & 0.0068 & 0.0020 & 0.0081 & 0.0180 & age& 0.0065\\
			$M_2$ & 0.0067 & 0.0020 & 0.0073 & 0.0178 & age&0.0071\\
			$M_3$ & 0.0100 & 0.0020 & 0.0025 & 0.0094 &  (age, gender)&0.0070\\
			$M_4$ & 0.0100 & 0.0024 & 0.0024 & 0.0094 &  (age, gender)&0.0074\\ \hline
\multicolumn{7}{c}{\underline{Step 4: selection of $X_*$  after $M_1$, $\hat Z_* $ = age, $\hat U $ = (bmi,gender) are selected}} \\
 \multicolumn{2}{c}{$X_*$} \vspace{1mm} & age & (bmi,age) & (gender,age) & (bmi,gender,age) & \\
  \multicolumn{2}{c}{PVC$(\hat Z_*)$}& 0.0108 & 0.0116 & 0.0116 & 0.0065 &  \\ \hline
  \multicolumn{7}{c}{\underline{Final model fitting after $M_1$, $\hat Z_* $ = age, $\hat U $ = (bmi,gender), $\hat X_*$ = (bmi,age,gender) are selected} }\\
\multicolumn{2}{c}{Parameter} & $\beta_0$: intercept& $\beta_1$: bmi & $\beta_2$: age & $\beta_3$: gender \\
\multicolumn{2}{c}{Estimate  (SE)} &
7.61 (4.90) & 0.75 (0.09) & -11.62 (0.87) & 0.21 (0.08) \\
  \hline
 \multicolumn{7}{l}{\small $\text{PVC}(\hat Z_*)$: PVC in \citet{Fang2016} with selected $\hat Z_*$ as instrument.}\\
% &&&& variance & 22.70 & 9.90 \\
		\end{tabular}
\end{table}

\begin{table}[htbp]
	\centering
	\caption{Instrument and variable selection and estimation of $\theta$ and $E(Y)$ in ACTG175 dataset\label{data:ACTG}} \vspace{2mm}
		\begin{tabular}{ccccccccccc}
			\hline
Arm	&	& $X_1$ & $X_2$ & $X_3$ & $X_4$ & $X_5$ & $X_6$ & Intercept &  $E(Y)$ & MAR \\
			\cline{3-11}
	1 &		$\vc $ & 0.0068 & 0.0105  &  0.0089 &   0.0085 &   0.0062 &   0.0195 & \\
	&		Est &  0&  0& 0.431   &   0.607 &   0&  0&   104.15    &     252.04 & 287.62 \\
	&		SE&     & &    0.098 &     0.097  & &&   34.40 &      10.722  & 9.287\\
	&	&	\multicolumn{9}{l}{Variable and instrument selected: $\hat U=\varnothing, \hat Z_*= \hat X_* = (X_3,X_4)$}\\
			\cline{3-11}
2 &			$\vc$ & 0.0074  &  0.0076  &  0.0092  &  0.0036  &  0.0030   & 0.0076 & \\
&			Est    &    0  &  0&    0.285 &     0.685 &  0&  0&     -39.22  &    336.20 & 341.25  \\
	&		SE&    &  & 0.071  &  0.070 & &   & 31.40 &     8.603 & 9.512\\
&	&		\multicolumn{9}{l}{Variable and instrument selected: $\hat U=\varnothing$, $\hat Z_* = \hat X_* =(X_3,X_4)$}\\
		\cline{3-11}%		\hline
3 &			$\vc$ & 0.0046  &  0.0068   & 0.0095  &  0.0099 &  0.0149 &  0.0111 & \\
	&		Est   &   0 &   1.239  &     0.491  &     0.645  &  0  &-0.083  &    -106.66  &   324.56  &354.85 \\
		&	SE&     &    0.698   &  0.098 &     0.116  & &    0.026  &   74.71  &     10.553  &9.441\\
&	&		\multicolumn{9}{l}{Variable and instrument selected:  $\hat U=\varnothing$, $\hat Z_*=\hat X_* =(X_2,X_3,X_4,X_6)$}\\
			\cline{3-11}%	\hline
	4 &		$\vc$ & 0.0060  & 0.0058  & 0.0069&   0.0088  &  0.0147&   0.0070 & \\
		&	Est   &  0  &0&   0.365 &     0.735 &0&0&   -82.38  &       319.38 &328.82  \\
		&	SE&  &  &   0.073  &  0.076 & &  & 27.72  &   7.460 & 9.512\\
	&	&		\multicolumn{9}{l}{Variable and instrument selected: $\hat U=\varnothing$, $\hat Z_* = \hat X_* =(X_3,X_4)$} \\ \hline
\multicolumn{11}{l}{\small Est: estimated regression coefficient or $E(Y)$} \\
\multicolumn{11}{l}{\small SE: standard error by bootstrap}\\
\multicolumn{11}{l}{\small MAR: estimated $E(Y)$ under missing at random}\\
	\end{tabular}
\end{table}

\end{document}